\documentclass[11pt]{article}
\usepackage[letterpaper,hmargin=1in,vmargin=1.5in]{geometry}
\usepackage{graphicx}
\usepackage{subfigure}
\usepackage[USenglish]{babel}
\usepackage{latexsym,amsthm,amsmath,amssymb,url}

\newtheorem{lemma}{Lemma}

\newtheorem{theorem}{Theorem}
\date{April 22, 2010}
\begin{document}
\newcommand{\2}{\vspace{0.3cm} }

\title{Betweenness Parameterized Above Tight Lower Bound\thanks{Part of this research has been supported by the EPSRC, grant EP/E034985/1, the IST Programme of the European Community, under the
PASCAL 2 Network of Excellence, and the Netherlands Organisation for Scientific Research (NWO), grant 639.033.403.}}

\author{
Gregory Gutin${}^1$, Eun Jung Kim${}^1$, Matthias Mnich${}^2$, and Anders Yeo${}^1$\\[6pt]
\small ${}^1$  Department of Computer Science\\[-3pt]
\small  Royal Holloway, University of London\\[-3pt]
\small Egham, Surrey TW20 0EX, UK\\[-3pt]
\small \texttt{\{gutin|eunjung|anders\}@cs.rhul.ac.uk}\\
\small ${}^2$ Technische Universiteit Eindhoven,
\small Eindhoven, The Netherlands\\[-3pt]
\small \texttt{m.mnich@tue.nl}
}




\maketitle

\begin{abstract}
We study ordinal embedding relaxations in the realm of parameterized complexity.
We prove the existence of a quadratic kernel for the {\sc Betweenness} problem parameterized above its tight lower bound, which is stated as follows.
For a set $V$ of variables and set $\mathcal C$ of constraints ``$v_i$~\mbox{is between}~$v_j$~\mbox{and}~$v_k$'', decide whether there is a bijection from $V$ to the set $\{1,\ldots,|V|\}$ satisfying at least $|\mathcal C|/3 + \kappa$ of the constraints in $\mathcal C$.
Our result solves an open problem attributed to Benny Chor
in Niedermeier's monograph ``Invitation to Fixed-Parameter Algorithms.''
The betweenness problem is of interest in molecular biology. An approach developed in this paper can be used to determine parameterized complexity of a number of other optimization problems on permutations parameterized above or below tight bounds.
\end{abstract}



\section{Introduction}
The problem of mapping points with measured pairwise distances into a target metric space has a long history and has been studied extensively from multiple perspectives due to its numerous applications.
The quality of such an embedding can be measured with various objectives; for example isometric embeddings preserve all distances while aiming at low-dimensional target spaces.
Yet, for many contexts in nearest-neighbor search, visualization, clustering and compression it is the order of distances rather than the distances themselves that captures the relevant information.
The study of such \emph{ordinal embeddings} dates back to the 1950's and has recently witnessed a surge in interest~\cite{AlonEtAl2008,BadoiuEtAl2008,BiluLinial2005,karpinski-2009}.
In an ordinal embedding the relative order between pairs of distances must be preserved as much as possible, i.e.,
one minimizes the relaxation of an ordinal embedding defined as the maximum ratio between two distances whose relative order is inverted by the embedding.

Here we study the one-dimensional ordinal embedding of partial orders that specify the maximum edge for some triangles.
This problem has been studied under the name {\sc Betweenness}~(see Part A12 of \cite{GareyJohnson1979}), which takes a set $V$ of variables and a set $\mathcal C$ of \emph{betweenness constraints} of the form $``v_i~\mbox{is between}~v_j~\mbox{and}~v_k$'' for distinct variables $v_i,v_j,v_k\in V$.
Such a constraint will be written as $(v_i,\{v_j,v_k\})$.
The objective is to find a bijection $\alpha$ from $V$ to the set $\{1,\hdots,|V|\}$ that ``satisfies'' the maximum number of constraints from $\mathcal C$, where a constraint $(v_i,\{v_j,v_k\})$ is \emph{satisfied} by~$\alpha$ if either $\alpha(v_j) < \alpha(v_i) < \alpha(v_k)$ or $\alpha(v_k) < \alpha(v_i) < \alpha(v_j)$ holds.
We also refer to~$\alpha$ as a \emph{linear arrangement} of~$V$.

Such linear arrangements are of significant interest in molecular biology, where for example markers on a chromosome need to be totally ordered as to satisfy the maximum number of constraints~\cite{CoxEtAl1990,GossHarris1975}.
More theoretical interest comes from the constraint programming framework with unbounded domains and interval graph recognition~\cite{KrokhinEtAl2004}.

Despite its simple formulation, the {\sc Betweenness} problem has a challenging computational complexity.
Already deciding if \emph{all} constraints can be satisfied by some linear arrangement is an $\mathsf{NP}$-complete problem \cite{Opatrny1979}.
Hence, the maximization problem is $\mathsf{NP}$-hard.
On the other hand, on the other hand, the average number of constraints satisfied by a uniformly random permutation of the variables is one-third of all constraints, and this fraction is tight.
Better approximation ratios are hard to achieve: the fraction of one-third is best-possible under the Unique Games Conjecture~\cite{CharikarEtAl2009}, and it is $\mathsf{NP}$-hard to find a linear arrangement that satisfies a $1-\epsilon$ fraction of the constraints for any $\epsilon \in (0,1/48)$~\cite{ChorSudan1998}.
The mere positive result is a polynomial-time algorithm that, assuming that some linear arrangement satisfies all constraints, finds a linear arrangement satisfying
at least half of them \cite{ChorSudan1998,Makarychev2009}.

A \emph{parameterized problem} is a subset $L\subseteq \Sigma^* \times
\mathbb{N}$ over a finite alphabet $\Sigma$. $L$ is
\emph{fixed-parameter tractable} if the membership of an instance
$(x,\kappa)$ in $\Sigma^* \times \mathbb{N}$ can be decided in time
$|x|^{O(1)} \cdot f(\kappa)$ where $f$ is a computable function of the
parameter $\kappa$~\cite{DowneyFellows1999,FlumGrohe2006,Niedermeier2006}.
If $L$ is $\mathsf{NP}$-hard, then the function $f(\kappa)$ must be superpolynomial
provided $\mathsf{P}\neq \mathsf{NP}$. Often $f(\kappa)$ is ``moderately exponential,''
which makes the problem practically feasible for small values of~$\kappa$.
Thus, it is important to parameterize a problem in such a way that the
instances with small values of $\kappa$ are of real interest.

One can parameterize the {\sc Betweenness} problem in the standard way by asking whether there exists a linear arrangement that satisfies at least $\kappa$ of the constraints, where $\kappa$ is the parameter.
This parameterized problem is trivially fixed-parameter tractable for the following simple reason.
A uniformly random permutation of the variables in $V$ satisfies one-third of the constraints in expectation; thus if $\kappa\leq |\mathcal C|/3$ then the answer is ``yes'' whereas if $\kappa > |\mathcal C|/3$ then the instance size is bounded by a function of $\kappa$ and we can solve the problem by brute-force.
Note that the standard parameterization is of little value as the parameter $\kappa$ will often be large.
Thus, it makes sense to consider the following natural parameterization above a tight lower bound of the {\sc Betweenness} problem~\cite{Niedermeier2006}.

Observe that $|\mathcal C|/3$ is a  lower bound on any optimal solution.
On the other hand, for a set $\mathcal C$ of constraints containing all three possible constraints on each 3-set of variables, no more than $|\mathcal C|/3$ of the constraints in $\mathcal C$ can be satisfied in any linear arrangement.
Hence the lower bound of one-third on the fraction of satisfiable constraints is tight, in the sense that it is attained by an infinite family of instances.
So the right question  to ask  is whether there exists a linear arrangement that satisfies at least $|\mathcal C|/3 + \kappa$ of the constraints.
The parameterized complexity of this problem attributed to Benny Chor is open, and was stated as such by Niedermeier \cite{Niedermeier2006}.
Since the {\sc Betweenness} problem is $\mathsf{NP}$-complete, the complementary question of whether all but $\kappa$ constraints are satisfiable by some linear arrangement is not fixed-parameter tractable, unless $\mathsf{P} = \mathsf{NP}$.
For the special case of a \emph{dense} set of constraints, containing a constraint for each 3-subset of variables, subexponential fixed-parameter algorithms were recently obtained~\cite{AlonEtAl2010,karpinski-2009}.

Mahajan and Raman \cite{MahajanRaman1999} were the first to consider
problems parametrized above tight lower bounds (PATLB). They indicated that such parameterizations are often the
only ones of practical value. Mahajan et al. \cite{MahajanEtAl2009} proved several
results for problems PATLB, and noted that only a few such problems were
investigated in the literature (partially, because they are often highly nontrivial to
study) and stated several open questions on fixed-parameter tractability of such
problems. Until very recently there were only three other papers on problems parameterized above tight lower bounds: Gutin et al. \cite{GutinEtAl2007}, Gutin et al. \cite{GutinEtAl2008}, and Villanger et al. \cite{Vill}.

Two recent papers by Alon et al. \cite{AlonEtAl2010} and Gutin et al. \cite{GutinEtAl2009}
solved two open questions in \cite{MahajanEtAl2009}, but several others remain open. A prominent example is to decide whether a planar graph $G$ contains an independent set of size $|V(G)|/4 + \kappa$; it is unknown whether this problem is fixed-parameter tractable or not. Another important problem is {\sc Max Lin-2} PATLB; the parameterized complexity of this problem remains open despite remains open despite extensive efforts \cite{CrowstonIPL,CrowstonSWAT,GutinEtAl2009}.

In this paper we settle Benny Chor's question~\cite[p. 43]{Niedermeier2006} about the parameterized complexity of the following problem:

\begin{quote}
  \textsc{Betweenness Above Tight Lower Bound} (\textsc{BATLB})\nopagebreak

  \emph{Instance:} a set $\mathcal C$ of betweenness constraints over variables $V$ and an integer $\kappa \geq 0$.\nopagebreak

  \emph{Parameter:} The integer $\kappa$.\nopagebreak

  \emph{Question:} Is there a bijection $\alpha:V\rightarrow\{1,\hdots,|V|\}$ that satisfies at least $|\mathcal C|/3+\kappa$ constraints from $\mathcal C$, that is, for  at least $|\mathcal C|/3+\kappa$ constraints $(v_i,\{v_j,v_k\})\in \mathcal C$ we have either $\alpha(v_j) < \alpha(v_i) < \alpha(v_k)$ or $\alpha(v_k) < \alpha(v_i) < \alpha(v_j)$?
\end{quote}

Our main result is that {\sc BATLB} is fixed-parameter tractable.
Moreover, we show that {\sc BATLB} has a \emph{kernel} of quadratic size, namely, any instance is polynomial-time reducible to an equivalent instance of size $O(\kappa^2)$.
(We give a formal definition of a kernel in the next paragraph.)
The kernel is obtained via a nontrivial extension of the recently introduced probabilistic \emph{Strictly Above/Below Expectation Method (SABEM)}~\cite{GutinEtAl2009}, which shows fixed-parameter tractability of {\sc Linear Ordering} and three special cases of {\sc Max Lin-2} parameterized above tight lower bounds.  Alon et al. \cite{AlonEtAl2010} further developed SABEM to prove fixed-parameter tractability of {\sc Max $r$-SAT} parameterized above a tight lower bound, but a simple modification of SABEM in \cite{AlonEtAl2010} is not applicable to {\sc BATLB}, see Section \ref{sec:strictly_above_method}.

We describe SABEM briefly in Section \ref{sec:strictly_above_method} and point out how to extend it in order to obtain a quadratic kernel for {\sc BATLB}.
The necessity to extend SABEM lies with the fact that a feasible solution to {\sc BATLB} is a permutation of variables; see Section \ref{sec:strictly_above_method} for a ``high-level'' discussion and Section \ref{sec:main} for details. Note that our extension of SABEM can be used to determine parameterized complexity of a number of other optimization problems on permutations parameterized above tight bounds, cf. \cite{gutin-2010,GuttmannM06}.


Given a parameterized problem $L$,
a \emph{kernelization of $L$} is a polynomial-time
algorithm that maps an instance $(x,\kappa)$ to an instance $(x',\kappa')$, the
\emph{kernel}, such that (i)~$(x,\kappa)\in L$ if and only if
$(x',\kappa')\in L$, (ii)~ $\kappa'\leq f(\kappa)$, and (iii)~$|x'|\leq g(\kappa)$ for some
functions $f$ and $g$. The function $g(\kappa)$ is called the {\em size} of the kernel.
A parameterized problem is fixed-parameter
tractable if and only if it is decidable and admits a
kernelization~\cite{FlumGrohe2006}; however, the kernels
obtained by this general result have impractically large size.
Therefore, one tries to develop kernelizations that yield problem
kernels of smaller size\footnote{Kernels are of great practical importance when it comes to solving $\mathsf{NP}$-hard problems; they can be interpreted as polynomial-time preprocessing with a quality guarantee~\cite{GuoNiedermeier2007,GutinK2009}.}. A survey of Guo and Niedermeier
\cite{GuoNiedermeier2007} on kernelization lists some problems for which
polynomial size kernels were obtained.
However, polynomial size kernels are known only for some fixed-parameter tractable problems and
Bodlaender et al. \cite{Bodlaender2008} proved that
many fixed-parameter tractable problems do not have polynomial size kernels under reasonable complexity-theoretic assumptions.

\section{Strictly Above/Below Expectation Method}
\label{sec:strictly_above_method}
The Strictly Above/Below Expectation Method \cite{GutinEtAl2009} is a way to prove fixed-parameter tractability of maximization (minimization, respectively) problems $\mathrm{\Pi}$ parameterized above (below, respectively) tight lower (upper, respectively) bounds. In that method, we first apply some reduction rules to reduce the given problem $\mathrm{\Pi}$ to its special case $\mathrm{\Pi}'.$
Then we introduce a random variable $X$ such that the answer to $\mathrm{\Pi}'$  is {\sc Yes} if and only if $X$ takes with positive probability a value greater than or equal to the parameter~$\kappa$.
If $X$ happens to be a symmetric random variable then the simple inequality $\mathbb P(X \geq \sqrt{\mathbb{E}[X^2]}) > 0$ can be useful;
here $\mathbb P(\cdot)$ and $\mathbb E[\cdot]$ denote probability and expectation, respectively.
An application is the {\sc Linear Ordering} problem and a special case of {\sc Max Lin-2} parameterized above tight lower bounds \cite{GutinEtAl2009}.
If $X$ is not symmetric then the following lemma can be used instead.
\begin{lemma}[Alon et al.~\cite{AlonEtAl2010}]
\label{lem:moments}
  Let $X$ be a real random variable and suppose that its first, second and fourth moments satisfy $\mathbb{E}[X] = 0$, $\mathbb{E}[X^2] = \sigma^2 > 0$ and $\mathbb{E}[X^4] \leq c \sigma^4$, respectively, for some constant $c$.
  Then $\mathbb P(X > \frac{\sigma}{2 \sqrt c}) > 0$.
\end{lemma}

\noindent
We combine this result with the following result from harmonic analysis.

\begin{lemma}[Hypercontractive Inequality~\cite{Bonami1970,Gross1975}]
\label{lem:polynomial}
  Let $f = f(x_1,\ldots,x_n)$ be a polynomial of degree $r$ in $n$ variables $x_1,\ldots,x_n$ each with domain $\{-1,1\}$.
  Define a random variable $X$ by choosing a vector $(\epsilon_1,\ldots,\epsilon_n)\in \{-1,1\}^n$ uniformly at random and setting $X = f(\epsilon_1,\ldots,\epsilon_n)$.
  Then $\mathbb E[X^4]\leq 9^r\mathbb E[X^2]^2$.
\end{lemma}

These two lemmas were used in \cite{AlonEtAl2010} and \cite{GutinEtAl2009} to prove fixed-parameter tractability of {\sc Max $r$-SAT} and three special cases of {\sc Max Lin-2} parameterized above tight lower bounds. Unfortunately, it appears to be impossible to introduce a random variable $X$ for which $\mathbb P(X \ge k)>0$ if and only if the answer to {\sc BATLB} is {\sc Yes} and such that $X$ is either symmetric or satisfies the conditions of Lemma \ref{lem:moments}. Thus, in the next section, we introduce $X$ for which we have a weaker property with respect to {\sc BATLB}: if $\mathbb P(X \ge k)>0$ then the answer to {\sc BATLB} is {\sc Yes}. This $X$, however, satisfies the conditions of Lemma \ref{lem:moments}.

To apply Lemma \ref{lem:moments}, we need to evaluate $\mathbb{E}[X^2]$.
While such evaluations by Alon et al. \cite{AlonEtAl2010} and Gutin et al. \cite{GutinEtAl2009} are rather straightforward, our evaluation of $\mathbb{E}[X^2]$ is quite involved
and requires assistance of a computer. Note that we cannot algebraically simplify $X$ as was done by Alon et al. \cite{AlonEtAl2010} in order to simplify computation of $\mathbb{E}[X^2]$. This is due to the weaker property of $X$ with respect to {\sc BATLB}: for some {\sc Yes}-instances of {\sc BATLB} we may have $\mathbb P(X \ge k)=0.$

\section{A quadratic kernel for BATLB}\label{sec:main}

We will now show fixed-parameter tractability of {\sc BATLB}.
In fact, we will prove a stronger statement, i.e., that this problem has a kernel of quadratic size.

For a constraint $C$ of $\mathcal C$ let $vars(C)$ denote the set of variables in $C$. We call a triple $A,B,C$ of distinct betweenness constraints {\em complete} if $vars(A)=vars(B)=vars(C).$

Consider the following {\em reduction rule}: if $\mathcal C$ contains a complete triple of constraints, delete these constraints from $\mathcal C$ and delete from $V$ any variable that appears only in the triple. Since for every linear arrangement exactly one constraint in each complete triple is satisfied we have the following:

\begin{lemma}\label{thm:reduce_3sets}
Let $(V,{\mathcal C},\kappa)$ be an instance of {\sc BATLB} and let $(V',{\mathcal C}',\kappa)$ be obtained from $(V,{\mathcal C},\kappa)$ by applying the reduction rule as long as possible. Then $(V,{\mathcal C},\kappa)$ is a {\sc Yes}-instance of {\sc BATLB} if and only if so is $(V',{\mathcal C}',\kappa)$.
\end{lemma}

An instance $(V,{\mathcal C},\kappa)$ of {\sc BATLB} is {\em irreducible} if it does not contain a complete triple. Observe that using Lemma \ref{thm:reduce_3sets} we can transform any instance into an irreducible one in $O(m^3)$ time. 

Consider an instance $(V,\mathcal C,\kappa)$, for a set $V$ of variables and a set $\mathcal C=\{C_1,\ldots ,C_m\}$ of  betweenness constraints, and
a random function $\phi:\ V\rightarrow \{0,1,2,3\}$. (The reason we consider a random function $\phi:\ V\rightarrow \{0,1,2,3\}$ rather than a random function $\phi:\ V\rightarrow \{0,1\}$ is given in the end of this section.) Let $\ell_i(\phi)$ be the number of variables in $V$ mapped by $\phi$ to $i$ for $i=0,1,2,3.$ Now obtain a bijection $\alpha:\ V\rightarrow\{1,\hdots,|V|\}$ by randomly assigning values $1,\ldots ,\ell_0(\phi)$ to all $\alpha(v)$ for which $\phi(v)=0$, and values $\sum_{i=0}^{j-1}\ell_{i}(\phi)+1,\ldots ,\sum_{i=0}^{j}\ell_{i}(\phi)$ to all $\alpha(v)$ for which $\phi(v)=j$ for every $j=1,2,3$. We call such a linear arrangement $\alpha$ a $\phi$-{\em compatible} bijection. It is easy to see that $\alpha$ obtained in this two stage process is, in fact, a random linear arrangement, but this fact is not going to be used here.

Now assume that a function $\phi:\ V\rightarrow\{0,1,2,3\}$ is fixed and consider a constraint $C_p = (v_i,\{v_j,v_k\})\in \mathcal C$. Let $\alpha$ be a random $\phi$-compatible bijection and $\nu_p(\alpha)=1$ if $C_p$ is satisfied and 0, otherwise. Let $w(C_p,\phi)=\mathbb{E}(\nu_p(\alpha))-1/3$ and $w({\mathcal C},\phi)=\sum_{p=1}^m w(C_p,\phi).$

\begin{lemma}\label{lem:yes}
If $w({\mathcal C},\phi)\ge \kappa$ then $(V,\mathcal C)$ is a {\sc Yes}-instance of {\sc BATLB}.
\end{lemma}
\begin{proof}
By linearity of expectation, $w({\mathcal C},\phi)\ge \kappa$ implies $\mathbb{E}(\sum_{p=1}^m\nu_p(\alpha))\ge m/3 + \kappa$. Thus, if $w({\mathcal C},\phi)\ge \kappa$ then there is a $\phi$-compatible bijection $\alpha$ that satisfies at least $m/3 + \kappa$ constraints. 
\end{proof}

Let $X=w({\mathcal C},\phi)$ and $X_p=w(C_p,\phi),$ $p=1,\ldots ,m.$ Observe that if $\phi$ is a random function from $V$ to $\{0,1,2,3\}$ then $X,X_1,\ldots ,X_m$ are random variables. Recall that $X=\sum_{p=1}^mX_p.$

\begin{lemma}\label{lem:EX=0}
We have $\mathbb{E}[X]=0.$
\end{lemma}
\begin{proof}
Let $C_p = (v_i,\{v_j,v_k\})\in \mathcal C$. Let us first find the distribution of $X_p$. It is easy to check that the probability that $\phi(v_i)=\phi(v_j)=\phi(v_k)$ equals 1/16 and  $X_p=0$ in such a case. The probability that $\phi(v_i)\neq\phi(v_j)=\phi(v_k)$ equals 3/16 and $X_p=-1/3$ in such a case.
The probability that $\phi(v_i)$ equals one of the non-equal $\phi(v_j), \phi(v_k)$ is equal to 6/16 and $X_p=1/6$ in such a case.
Now suppose that $\phi(v_i),\phi(v_j)$ and $\phi(v_k)$ are all distinct. The probability that $\phi(v_i)$ is between $\phi(v_j)$ and $\phi(v_k)$ is 2/16 and $X_p=2/3$ in such a case.
Finally, the probability that $\phi(v_i)$ is not between $\phi(v_j)$ and $\phi(v_k)$ is 4/16 and $X_p=-1/3$ in such a case. Now we can give the distribution of $X_p$ in Table 1.

\2

\begin{table}
\label{tab:Xp_dist}
\centering
\begin{tabular}{|c|c|c|c|}\hline
$|\{\phi(v_i),\phi(v_j), \phi(v_k)\}|$ & Relation & Value of $X_p$ & Prob.\\ \hline
1 & $\phi(v_i) = \phi(v_j) = \phi(v_k)$ & 0 & 1/16 \\
2 & $\phi(v_i)\neq \phi(v_j)=\phi(v_k)$ & $-1/3$ & 3/16\\
2 & $\phi(v_i)\in \{\phi(v_j), \phi(v_k)\}$ & $1/6$ & 6/16 \\
3 & $\phi(v_i)$ is between $\phi(v_j)$ and $\phi(v_k)$ & $2/3$ & 2/16\\
3 & $\phi(v_i)$ is not between $\phi(v_j)$ and $\phi(v_k)$ & $-1/3$ & 4/16\\
\hline
\end{tabular}
\vspace{0.2cm}
\caption{Distribution of $X_p$.}
\end{table}

\2

Using this distribution, it is easy to see that $\mathbb{E}[X_p]=0$ and, thus, $\mathbb{E}[X]=\sum_{p=1}^m\mathbb{E}[X_p]=0.$ \end{proof}

\begin{lemma}\label{lem:poly}
The random variable $X$ can be expressed as a polynomial of degree 6 satisfying the conditions of
Lemma \ref{lem:polynomial}.
\end{lemma}
\begin{proof}
Consider $C_p = (v_i,\{v_j,v_k\})\in \mathcal C$. Let $\epsilon^i_1=-1$ if $\phi(v_i)=0$ or 1 and $\epsilon^i_1=1$, otherwise. Let $\epsilon^i_2=-1$ if $\phi(v_i)=0$ or 2 and $\epsilon^i_2=1$, otherwise. Similarly, we can define $\epsilon^j_1,\epsilon^j_2,\epsilon^k_1,\epsilon^k_2.$ Now $\epsilon^i_1\epsilon^i_2$ can be seen as a binary representation of a number from the set $\{0,1,2,3\}$ and $\epsilon^i_1\epsilon^i_2\epsilon^j_1\epsilon^j_2\epsilon^k_1\epsilon^k_2$ can be viewed as a binary representation of a number from the set $\{0,1,\ldots ,63\}$, where $-1$ plays the role of 0.

We can write $X_p$ as the following polynomial: $$\frac{1}{64}\sum_{q=0}^{63} (-1)^{s_q}w_q\cdot  (\epsilon^i_1+c^{iq}_1)(\epsilon^{i}_2+c^{iq}_2)(\epsilon^{j}_1+c^{jq}_1)(\epsilon^{j}_2+c^{jq}_2)(\epsilon^{k}_1+c^{kq}_1)(\epsilon^{k}_2+c^{kq}_2),$$ where $c^{iq}_1c^{iq}_2c^{jq}_1c^{jq}_2c^{kq}_1c^{kq}_2$ is the binary representation of $q$, $s_q$ is the number of digits equal $-1$ in this representation, and $w_q$ equals the value of $X_p$ for the case when the binary representations of $\phi(v_i),\phi(v_j)$ and $\phi(v_k)$ are $c^{iq}_1c^{iq}_2$, $c^{jq}_1c^{jq}_2$ and $c^{kq}_1c^{kq}_2$, respectively. The actual values for $X_p$ for each case are given in the proof Lemma \ref{lem:EX=0}. The above polynomial is of degree 6. It remains to recall that $X=\sum_{p=1}^m X_p.$
\end{proof}

\begin{lemma}
\label{thm:secondmoment}
For an irreducible instance $(V,{\mathcal C},\kappa)$ of {\sc BATLB} we have $\mathbb{E}[X^2] \geq \frac{11}{768}m$.
\end{lemma}
\begin{proof}
First, observe that $\mathbb{E}[X^2] = \sum_{l=1}^m \mathbb{E}[X_l^2] + \sum_{1\leq l \not= l' \leq m}\mathbb{E}[X_lX_{l'}]$.
We will compute $\mathbb{E}[X_l^2]$ and $\mathbb{E}[X_lX_{l'}]$ separately.

Using the distribution of $X_l$ given in Table 1, it is easy to see that $\mathbb{E}[X_l^2] = 11/96=88/768$. It remains to show that

\begin{equation}\label{eq1}
\sum_{1\leq l\neq l' \leq m}\mathbb{E}[X_lX_{l'}] \geq - \frac{77}{768}m.
\end{equation}
Indeed, (\ref{eq1}) and $\mathbb{E}[X_l^2] = 88/768$ imply that
\begin{equation*}
\mathbb{E}[X^2] =    \sum_{l=1}^m \mathbb{E}[X_l^2] + \sum_{1\leq l\neq l' \leq m}\mathbb{E}[X_lX_{l'}]
                        \geq \frac{88}{768}m - \frac{77}{768}m
                        =    \frac{11}{768}m,
\end{equation*}

In the remainder of this proof we show that (\ref{eq1}) holds.
Let $C_l, C_{l'}$ be a pair of distinct constraints of $\mathcal{C}$.
To evaluate $\mathbb{E}[X_lX_{l'}]$, we consider several cases.
A simple case is when the sets $vars(C_l)$ and $vars(C_{l'})$ are disjoint: then $X_l$ and $X_{l'}$ are independent random variables and, thus, $\mathbb{E}[X_lX_{l'}] = \mathbb{E}[X_l]\mathbb{E}[X_{l'}] = 0$.
Let $U = \{(l,l')~|~C_l,C_{l'}\in\mathcal C, l \not= l'\}$ be the set of all ordered index pairs corresponding to distinct constraints in $\mathcal C$.
We will classify subcases of this case by considering some subsets of $U.$ Let
\begin{eqnarray*}
S_1(u) & = & \{(l,l')\in U\  :\ C_l = (u,\{a,b\}), C_{l'} = (u,\{c,d\}),\ a,b,c,d\in V\},\\
S_2(u) & = & \{(l,l')\in U\  :\ C_l = (a,\{u,b\}), C_{l'} = (c,\{u,d\}),\ a,b,c,d\in V\},\\
S_3(u) & = & \{(l,l'),(l',l)\in U\  :\ C_l = (u,\{a,b\}), C_{l'} = (c,\{u,d\}),\ a,b,c,d\in V\},\\
S_4(u,v) & = & \{(l,l')\in U\  :\ C_l = (u,\{v,a\}), C_{l'} = (u,\{v,b\}),\ a,b\in V\}\\
          & \cup & \{(l,l')\in U\  :\ C_l = (v,\{u,a\}), C_{l'} = (v,\{u,b\}),\ a,b\in V\},\\
S_5(u,v) & = & \{(l,l')\in U \  :\ C_l = (a,\{u,v\}), C_{l'} = (b,\{u,v\}),\ a,b\in V\},\\
S_6(u,v) & = & \{(l,l'),(l',l) \in U\  :\ C_l = (u,\{v,a\}), C_{l'} = (b,\{u,v\}),\ a,b\in V\}\\
          & \cup & \{(l,l'),(l',l)\in U\  :\ C_l = (v,\{u,a\}), C_{l'} = (b,\{u,v\}),\ a,b\in V\},\\
S_7(u,v) & = & \{(l,l'),(l',l) \in U\  :\ C_l = (u,\{v,a\}), C_{l'} = (v,\{u,b\}),\ a,b\in V\}\\
S_8(u,v,w) & = & \{(l,l')\in U\  :\ vars(C_l)=vars(C_{l'})=\{u,v,w\}\}.
\end{eqnarray*}

Let $u,v\in V$ be a pair of distinct variables. Observe that
$S_4(u,v)=(S_1(u)\cap S_2(v))\cup (S_1(v)\cap S_2(u)),$
$S_5(u,v)=S_2(u)\cap S_2(v),$
$S_6(u,v)=(S_3(u)\cap S_2(v))\cup (S_3(v)\cap S_2(u))$ and
$S_7(u,v)=S_3(u)\cap S_3(v)$.
Let $u,v,w\in V$ be a triple of distinct variables. Observe that
\begin{equation}\label{eq2}
S_8(u,v,w)= (S_3(u)\cap S_3(v)\cap S_2(w))\cup (S_3(v)\cap S_3(w)\cap S_2(u)) \cup (S_3(w)\cap S_3(u)\cap S_2(v)).
\end{equation}

For a variable $u\in V$, let $b(u)=|\{l\  :\ C_l = (u,\{a,b\}), a,b\in V\}|$ and $e(u)=|\{l\  :\ C_l = (a,\{u,b\}), a,b\in V\}|.$ Observe that $|S_1(u)|=b(u)(b(u)-1)$, $|S_2(u)|=e(u)(e(u)-1)$ and
$|S_3(u)|=2b(u)e(u)$.

For a pair $u,v\in V$, let $c^u_v=|\{l\  :\ C_l = (u,\{v,a\}), a\in V\}|$ and $c_{uv}=|\{l\  :\ C_l = (a,\{u,v\}), a\in V\}|.$ Observe that
$|S_4(u,v)|=c^u_v(c^u_v-1) + c^v_u(c^v_u-1),$
$|S_5(u,v)|=c_{uv}(c_{uv}-1),$
$|S_6(u,v)|= 2(c^u_v + c^v_u)\cdot c_{uv}$ and
$|S_7(u,v)|=2c^u_vc^v_u.$
Let $u,v,w\in V$ be a triple of distinct variables. Since $\mathcal C$ is irreducible, the number of ordered pairs $(C_l,C_{l'})$ for which $vars(C_l) = vars(C_{l'})=\{u,v,w\}$ is at most $2$, i.e., $|S_8(u,v,w)|\le 2$.

\begin{table}
\label{tab:set_types}
\centering
\renewcommand\arraystretch{1.4}
\begin{tabular}{|cccrr|}
\hline
 Set   & Union/intersection Form                           & $|$Set$|$                         & \mbox{  }$768\mathbb{E}[X_lX_{l'}]$ & $\mbox{  }768w'$ \\
\hline
     $S_1(u)$   &     --                    & $b(u)(b(u)-1)$                    & $12 = w_1$  &  $12$  \\
   $S_2(u)$     &     --                   & $e(u)(e(u)-1)$                    & $3 = w_2$  &  $3$  \\
$S_3(u)$     &         --                & $b(u)e(u) + e(u)b(u)$             & $-6 = w_3$ & $-6$  \\
$S_4(u,v)$ &   $(S_1(u)\cap S_2(v))\cup (S_1(v)\cap S_2(u))$                                      &   $c^u_v(c^u_v-1) + c^v_u(c^v_u-1)$ & $24 = w_4$  &  $9$ \\
$S_5(u,v)$  & $S_2(u)\cap S_2(v)$                   & $c_{uv}(c_{uv}-1)$                & $36 = w_5$  &  $30$ \\
$S_6(u,v)$ &   $(S_3(u)\cap S_2(v))\cup (S_3(v)\cap S_2(u))$     & $2(c^u_v + c^v_u)\cdot c_{uv}$    & $-18 = w_6$  & $-15$ \\
$S_7(u,v)$ &  $S_3(u)\cap S_3(v)$                 & $2c^u_vc^v_u$                     & $-6 = w_7$  & $6$ \\
$S_8(u,v,w)$  & see (\ref{eq2})                             & $\leq 2$                               & $-44 = w_8$ & $-11$ \\
\hline
\end{tabular}\\
\vspace{0.2cm}
\caption{Data for sets $S_i(\dot)$, $i=1,2,\ldots ,8$.}
\end{table}

We list the sets $S_i(\cdot )$, their union/intersection forms (for $i=4,5,6,7$) and their sizes in Table~2.
If $(l,l')$ belongs to some $S_i$ but to no $S_j$ for $j>i$, then Table~2 also contains the value $768\cdot \mathbb{E}[X_lX_{l'}]$,
in the row corresponding to $S_i$. These values cannot be easily calculated analytically as there are many cases to consider and we have calculated them using a computer.
We will briefly describe how our program computes $\mathbb{E}[X_lX_{l'}]$  using as an example the case $(l,l')\in S_1(u)$, i.e., $ C_l = (u,\{a,b\}), C_{l'} = (u,\{c,d\}).$ For each $(q_1,q_2,q_3,q_4,q_5)\in \{0,1,2,3\}^5$ the probability of $(u,a,b,c,d)=(q_1,q_2,q_3,q_4,q_5)$ is $4^{-5}$ and the corresponding value of $X_lX_{l'}$ can be found in Table 2.

We are now ready to compute a lower bound on the term $\sum_{1\leq l\neq l' \leq m}\mathbb{E}[X_lX_{l'}]$.
Define the values $w_i'$ for $i=1,2,\ldots,8$ as it is done in Table~2. We will now show that the following holds
(note that the sets we sum over have to contain distinct elements).

\2

$\begin{array}{rcl}
\vspace{0.15cm}
 \sum_{1\leq l\neq l' \leq m}\mathbb{E}[X_lX_{l'}] & = & \sum_{u \in V} \sum_{i=1}^3 |S_i(u)| w_i'
 + \sum_{\{u,v\} \subseteq  V  }  \sum_{i=4}^7   |S_i(u,v)| w_i'  \\
 & & + \sum_{\{u,v,w\} \subseteq  V  } |S_8(u,v)| w_8'  \\
\end{array}$

\2


In order to show the above we consider the possible cases for $(l,l')\in U$.

\vspace{3mm}

\noindent
\textbf{Case 1:} $|vars(C_l)\cap vars(C_{l'})|= 0$.
In this case $\mathbb{E}[X_lX_{l'}] =0$ and the corresponding $(l,l')$ does not belong to any $S_i$ and therefore
contributes zero to the right-hand side above.

\vspace{3mm}

\noindent
\textbf{Case 2:} $|vars(C_l)\cap vars(C_{l'})|= 1$.
Each pair $(l,l') \in S_1(u)$ contributes $\frac{12}{768}$ to both sides of the above equation, as in this case
$(l,l')$ does not belong to any $S_j$ with $j>1$. Analogously if $(l,l')\in S_2(u)$ then it contributes  $\frac{3}{768}$
to both sides of the above equation. Furthermore if $(l,l')\in S_3(u)$ then it contributes $-\frac{6}{768}$.

\vspace{3mm}

\noindent
\textbf{Case 3:} $|vars(C_l)\cap vars(C_{l'})|= 2$.
Consider a pair $(l,l')\in S_4(u,v)$ and assume, without loss of generality, that $(l,l')\in S_1(u)\cap S_2(v)$.
Note that $(l,l')$ contributes $\frac{24}{768}$ to the left-hand side of the equation and it contributes $w_1' + w_2' + w_4' =
\frac{24}{768}$ to the right-hand side (as $(l,l')\in S_1(u)\cap S_2(v) \cap S_4(u,v)$).
Analogously if $(l,l')\in S_5(u,v)$ we get a contribution of $w_5 = \frac{36}{768} = w_2' + w_2' + w_5'$ to both
sides of the equation. If $(l,l')\in S_6(u,v)$ we get a contribution of $w_6 = -\frac{18}{768} = w_3' + w_2' + w_6'$ to both
sides of the equation. If $(l,l')\in S_7(u,v)$ we get a contribution of $w_7 = -\frac{6}{768} = w_3' + w_3' + w_7'$ to both
sides of the equation.

\vspace{3mm}

\noindent
\textbf{Case 4:} $|vars(C_l)\cap vars(C_{l'})|= 3$. Assume, without loss of generality, that $(l,l')\in S_3(u)\cap S_3(v)\cap S_2(w)$
and note that $(l,l') \in S_7(u,v) \cap S_6(u,w) \cap S_6(v,w)$. Therefore
we get a contribution of $w_8 = -\frac{44}{768} = w_3' + w_3' + w_2' + w'_7 + w'_6 + w'_6 + w_8'$ to both
sides of the equation.

\vspace{3mm}

\noindent
Therefore the above equation holds, which implies the following:

\2

\newcommand{\InsX}{\hspace{3.75cm} }
\newcommand{\InsY}{\hspace{2.4cm} }
\newcommand{\InsXX}{\hspace{2.9cm} }
$\begin{array}{l}
\vspace{0.2cm}
\sum_{1\leq l\neq l' \leq m}\mathbb{E}[X_lX_{l'}]  =   \sum_{u\in V}  \left(  |S_1(u)| w'_1 + |S_2(u)| w'_2 + |S_3(u)| w'_3 \right) \\
\vspace{0.7cm}
\InsX{}     + \sum_{\{u,v\} \subseteq V} \sum_{i=4}^7 |S_i(u,v)| w'_i        + \sum_{\{u,v,w\} \subseteq V } |S_8(u,v,w)| w'_8 \\
\vspace{0.2cm}
\InsY{}  =  \frac{1}{2\cdot 768}\sum_{u\in V}    \left( 6(2b(u)-e(u))^2 - 24b(u)-6e(u)\right) \\
\vspace{0.2cm}
\InsXX{}  + \frac{1}{2\cdot 768}\sum_{\{u,v\} \subseteq V}  \left( 15(c^u_v + c^v_u - 2c_{uv})^2 + 12\left(\frac{c^u_v - c^v_u}{2}\right)^2 - 18(c^u_v + c^v_u) - 60c_{uv}\right) \\
\InsXX{} +              \sum_{\{u,v,w\} \subseteq V} |S_8(u,v,w)| w'_8 \\
\end{array}$

\2
\2

To complete the proof of the lemma it remains to translate this sum into a function on the number of constraints.
In that respect, notice that $\sum_{u\in V}b(u) = m$ and $\sum_{u\in V}e(u) = 2m$.
Further, each clause $(u,\{v,w\})$ contributes exactly one unit to each of $c^u_v$ and $c^u_w$, as well as exactly one unit to $c_{vw}$.
Hence $\sum_{\{u,v\}\subset V}(c^u_v + c^v_u)= 2m$ and $\sum_{\{u,v\}\subset V}c_{uv} = m$. Since $\mathcal C$ is irreducible, the number of ordered pairs $(C_l,C_{l'})$ for which $vars(C_l) = vars(C_{l'})$ is at most $m/2$ and, thus, $$\sum_{\{u,v,w\} \subseteq V}  |S_8(u,v,w)| w'_8\le m\cdot w'_8.$$

Together these bounds imply that
\begin{equation*}
\sum_{1\leq l\neq l' \leq m}\mathbb{E}[X_lX_{l'}] \geq - \frac{36}{2\cdot 768}m - \frac{96}{2\cdot 768}m - \frac{11}{768}m
                                                                  =    - \frac{77}{768}m \enspace
\end{equation*}
and (\ref{eq1}) holds.
\end{proof}

We are now ready to prove the main result.
\begin{theorem}\label{thmain}
{\sc BATLB} has a kernel of size $O(\kappa^2).$
\end{theorem}
\begin{proof}
Let $(V,\mathcal C)$ be an instance of {\sc BATLB}. By Lemma \ref{thm:reduce_3sets}, in time $O(m^3)$ we can obtain an irreducible instance $(V',\mathcal C')$ such that $(V,\mathcal C)$ is a {\sc Yes}-instance if and only if $(V',{\mathcal C}')$ is a {\sc Yes}-instance. Let $m'=|{\mathcal C}'|$ and
let  $X$ be the random variable defined above. Then $X$ is expressible as a polynomial of degree 6 by Lemma \ref{lem:poly}; hence it follows from Lemma~\ref{lem:polynomial}  that $\mathbb{E}[X^4] \leq 9^{6} \mathbb{E}[X^2]^2$. Consequently, $X$ satisfies the conditions of Lemma~\ref{lem:moments}, from which we conclude in combination with Lemma~\ref{thm:secondmoment} that $\mathbb P\left(X > \frac{1}{2\cdot 9^{3}} \sqrt{\frac{11}{768}m'}\right) > 0$.
By Lemma \ref{lem:yes} if $\frac{1}{2\cdot 9^{3}} \sqrt{\frac{11}{768}m'}\ge \kappa$ then $(V',{\mathcal C}')$ is a {\sc Yes}-instance for {\sc BATLB}.
Otherwise, we have $m'= O(\kappa^2)$.
This concludes the proof of the theorem.
\end{proof}

We complete this section by answering the following natural question: why have we considered functions $\phi:\ V\rightarrow\{0,1,2,3\}$ rather than functions $\phi:\ V\rightarrow\{0,1\}$? The latter would involve less computations and give a smaller degree of the polynomial representing $X$. The reason is that our proof of Lemma \ref{thm:secondmoment} would not work for functions $\phi:\ V\rightarrow\{0,1\}$ (we would only be able to prove that $\mathbb{E}[X^2] \geq \sum_{\{u,v\}\subset V}[c^v_u+c^u_v-2c_{uv}]^2$, which is not enough).


\end{document}